\documentclass[submission,copyright,creativecommons]{eptcs}
\providecommand{\event}{Express/SOS 2018} 
\usepackage{underscore}           

\usepackage{amsthm}
\usepackage{amssymb}
\usepackage{amsmath}

\usepackage{enumerate}
\usepackage{paralist}

\usepackage{listings}

\theoremstyle{plain}
\newtheorem{theorem}{Theorem}
\newtheorem{lemma}{Lemma}

\theoremstyle{remark}

\theoremstyle{definition}
\newtheorem{defn}{Definition}[section]

\newcounter{example}
\newenvironment{example}[1][]{\refstepcounter{example}\par\medskip
   \noindent \textbf{Example~\theexample. #1}\rmfamily}{\medskip}
 
\allowdisplaybreaks

\newcommand{\tran}[1]{\stackrel{#1}{\rightarrow}}
\newcommand{\revtran}[1]{\stackrel{#1}{\rightsquigarrow}}

\title{Reversing Parallel Programs with Blocks and Procedures}
\author{
James Hoey \qquad\qquad Irek Ulidowski
\institute{Department of Informatics\\
University of Leicester, UK}
\email{\quad jbh11@leicester.ac.uk \quad\qquad iu3@leicester.ac.uk}
\and
Shoji Yuen \institute{Graduate School of Informatics \\ Nagoya University, Japan}
\email{yuen@i.nagoya-u.ac.jp}
}

\def\titlerunning{Reversing Parallel Programs with Blocks and Procedures}
\def\authorrunning{J. Hoey, I. Ulidowski \& S. Yuen}

\begin{document}
\maketitle

\begin{abstract}
We show how to reverse a while language extended with blocks, local variables, procedures and the interleaving parallel composition. Annotation is defined along with a set of operational semantics capable of storing necessary reversal information, and identifiers are introduced to capture the interleaving order of an execution. Inversion is defined with a set of operational semantics that use saved information to undo an execution. We prove that annotation does not alter the behaviour of the original program, and that inversion correctly restores the initial program state. 
\end{abstract}

\section{Introduction} \label{sec-intro}
Reverse execution of programs is the ability to take the final program state of an execution, undo all effects that were the result of that execution, and restore the exact initial program state. This is a desirable capability as it has applications to many active research areas, including debugging \cite{EG2014} and Parallel Discrete Event Simulation \cite{RJ1990}. When combined with parallelism, reverse execution removes issues relating to non-deterministic execution orders, allowing specific execution interleavings to be analysed easily.

In our previous work \cite{JH2017}, we described a state-saving approach to reversible execution of an imperative \emph{while language}. Similarly to RCC \cite{KP2014}, we generated two versions of a program, the \emph{augmented forwards version} to save all necessary \emph{reversal information} alongside its execution, and the \emph{inverted version} that uses this saved data to undo all changes. We proved that augmentation did not alter the behaviour of our program, and that inversion correctly restores the initial program state. We then experimented with reversing a tiny language containing assignments and interleaving parallel composition.

In this paper, we extend the while language with blocks, local variables and procedures, as well as the parallel composition operator. Local variables mean we must recognise scope, for example different versions of a shared name used in parallel. Issues arise with the traditional approach, specifically with recursion, and calls to the same procedure executing in parallel. Annotation and inversion are defined, allowing this extended language to be executed forwards with state-saving, as well as in reverse using this saved information. The process of assigning identifiers to statements as we execute them is described, focusing on \emph{backtracking order}, where statements are undone in the inverted order of the forwards execution. We mention future work on causal-consistent reversibility \cite{VD2004,IL2014,IP2007} in the conclusion. 

Consider the example shown in Figure \ref{intro-ex-1}, where \texttt{w1.0} and $\lambda$ can be ignored. This is a simple model of a restaurant with two entrances. One where a single person is continually allowed to enter, increasing the number of current single guests (\texttt{c}), until the total capacity (\texttt{c} + \texttt{r}) reaches the maximum (\texttt{m}). The other allows a reserved group of two to enter, increasing the number of reserved guests (\texttt{r}). Let the initial state be that \texttt{m} = \texttt{4}, \texttt{c} = \texttt{0} and \texttt{r} = \texttt{0}. The execution begins with two full iterations of the while loop, allowing two people to enter meaning \texttt{c} = \texttt{2}. Next, the condition of the loop is evaluated, but the body is not yet executed. Interleaving now occurs, setting \texttt{r} to \texttt{2}. Finally, the body of the loop is now executed, before the condition evaluates to false and the loop finishes. The final state is \texttt{m} = \texttt{4}, \texttt{c} = \texttt{3} and \texttt{r} = \texttt{2}, which should be invalid as the total number of guests (\texttt{c} + \texttt{r}) $>$ \texttt{m}. This executed version of the annotated program is shown in Figure~\ref{intro-ex-2}, where each statement now has a stack populated with identifiers in the order in which the statement occurred (starting at 1). One solution to finding this bug is reverse execution. The inverted version we generate (which, coincidentally is identical to Figure \ref{intro-ex-2}) allows step-by-step reversal, using identifiers to remove non-determinism. Backtracking through this execution removes the difficulties of cyclic debugging where different interleavings can occur. Using this, we can see that we wrongly commit to allowing the third single person to enter, meaning the condition gave true when expected to give false. Examining this further, we can see that the reserved guests are not considered until they have arrived, meaning this condition is not aware that the maximum capacity is actually \texttt{m} - 2. This is an example of a \emph{race} between the writing and reading of \texttt{r}. This is fixed using the condition \texttt{((m-c-2-1) >= 0)}.

Our main contributions are 
\begin{enumerate}
\item  The definition of three sets of operational semantics for our language, namely for \emph{traditional} forwards only execution, annotated forwards execution, and for reverse execution.   
\item Annotation allows all necessary state-saving, and the use of identifiers to record the interleaving order of execution. Inversion then uses the saved information to reverse via backtracking order.
\item Results showing that annotation does not alter the behaviour of the original program, and that inversion correctly restores to the initial program state.
\end{enumerate}  
We also have a prototype simulator under development. This will be capable of implementing both the forwards and reverse execution, and used for both performance evaluation and validation of our results.

\begin{figure}[t]
  \begin{minipage}[b]{0.49\linewidth}
   \centering
   {\small \begin{lstlisting}[xleftmargin=2.0ex,mathescape=true]
 $\texttt{par \{}$
  $\texttt{while w1.0 ((m - c - r - 1) >= 0) do}$
   $\texttt{c = c + 1}$ $\lambda$;
  $\texttt{end}$ $\lambda$; $\texttt{\}}$     
 $\texttt{\{}$ $\texttt{r}$ = $\texttt{2}$ $\lambda$; $\texttt{\}}$
\end{lstlisting} }
    \caption{Original program}
	\label{intro-ex-1}
  \end{minipage}
  \begin{minipage}[b]{0.49\linewidth}
    \centering
    {\small \begin{lstlisting}[xleftmargin=2.0ex,mathescape=true]
 $\texttt{par \{}$
  $\texttt{while w1.0 ((m - c - r - 1) >= 0) do}$
   $\texttt{c = c + 1 (}$$\lambda \texttt{, [2,4,7])}$;
  $\texttt{end (}$$\lambda \texttt{, [1,3,5,8,9])}$; $\texttt{\}}$
 $\texttt{\{}$ $\texttt{r}$ = $\texttt{2 (}$$\lambda \texttt{, [6])}$; $\texttt{\}}$
	\end{lstlisting} }
    \caption{Executed annotated program}
    \label{intro-ex-2}
  \end{minipage}
\end{figure}

\subsection{Related Work}
Program inversion has been the focus of many works for many years, including Jefferson \cite{DJ1985}, Gries \cite{DG1981} and Gl\"{u}ck and Kawabe \cite{RG2004,RG2005}. The Reverse C Compiler (RCC) by Perumalla \cite{KP2014} describes a state-saving approach to reversibility of C programs.  The Backstroke framework \cite{GV2011} and extensions of it by Schordan et al \cite{MS2015} describe an approach to reversing C++ in the setting of Parallel Discrete Event Simulation \cite{RJ1990}.  The reversible programming language Janus,  worked on in \cite{TY2008A,TY2007}, adds additional information into the source code, making all programs reversible. More recent work on reversible imperative programs by Gl\"{u}ck and Yokoyama \cite{RG2016,RG2017} introduce the languages R-WHILE and R-CORE. Reversibility of algebraic process calculi is the focus of work by Phillips and Ulidowski \cite{IP2007,IP2012}, where the notion of identifiers was introduced. There has been work on reversible object oriented programming languages, including that of Schultz \cite{US2016} and the language ROOPL \cite{TH2017}. The application of reverse computation to debugging of message passing concurrent programs is considered by Giachino et al \cite{EG2014}.

\section{Programming Language, Environments and Scope} \label{sec-les}
Let \textbf{P} be the set of all programs and \textbf{S} be the set of all statements. Each program \texttt{P} will be either a statement \texttt{S}, the sequential composition of programs \texttt{P;Q} or the parallel composition of programs \texttt{P}~\texttt{par}~\texttt{Q} (sometimes written as \texttt{par} \texttt{\{P\}\{Q\}}). Each statement will either be a skip operation (empty statement), an assignment, a conditional, a loop, a block, a variable or procedure declaration, a variable or procedure removal or a call. A block consists of the declaration of both local variables \texttt{DV} and procedures \texttt{DP}, a body that uses these, and then the removal of local procedures \texttt{RP} and variables \texttt{RV}. Procedures do not have arguments, static scope is assumed and recursion is permitted. The syntax of this language is shown below, including arithmetic and boolean expressions. Note that the constructs \texttt{runC} and \texttt{runB} are reserved words that appear in our syntax, but not in original programs, and will be explained in Section \ref{sec-trad}. These allow static operational semantics to be defined, needed to aid state-saving and in our results.

Many statements contain a path \texttt{pa} that is explained in Section \ref{ssec-scope}. Each conditional, loop, block, procedure and procedure call statement has a unique identifier named \texttt{In}, \texttt{Wn}, \texttt{Bn}, \texttt{Pn} and \texttt{Cn} respectively, each of which is an element of the sets \textbf{In}, \textbf{Wn}, \textbf{Bn}, \textbf{Pn} and \textbf{Cn} respectively. The set union of these gives us the set of \emph{construct identifiers} \textbf{CI}. Note a procedure has a name from the set \textbf{n} (appears in code, and is potentially duplicated), as well as a unique identifier \texttt{Pn}. 
\begin{align*}
\texttt{P} &::= \texttt{$\varepsilon$} ~|~ \texttt{S} ~|~ \texttt{P; P} ~|~ \texttt{P par P} \\
\texttt{S} &::= \texttt{skip} ~|~  \texttt{X = E pa} ~|~ \texttt{if In B then P else Q end pa} ~| \\ &\phantom{\texttt{::=}}  \texttt{while Wn B do P end pa} ~|~ \texttt{begin Bn DV DP P RP RV end}  ~| \\ &\phantom{\texttt{::=}}\texttt{call Cn n pa} ~|~   \texttt{runC Cn P end} ~|~   \texttt{runB P end} 
\end{align*} 
\vspace{-0.8cm}
\begin{align*}
\texttt{DV} &::= \texttt{$\varepsilon$} ~|~ \texttt{var X = v pa; DV}   &
\texttt{DP} &::= \texttt{$\varepsilon$} ~|~ \texttt{proc Pn n is P pa; DP} \\
\texttt{RV} &::= \texttt{$\varepsilon$} ~|~ \texttt{remove X = v pa; RV}  &
\texttt{RP} &::= \texttt{$\varepsilon$} ~|~ \texttt{remove Pn n is P pa; RP} \\
\texttt{E}   &::=   \texttt{Var} ~|~ \texttt{n} ~|~  \texttt{(E)} ~|~ \texttt{E Op E}  &\texttt{B}  &::= \texttt{T} ~|~ \texttt{F} ~|~ \lnot\texttt{B}~|~ \texttt{(B)} ~|~ \texttt{E == E} ~|~ \texttt{E > E} ~|~  \texttt{B $\land$ B} 
\end{align*} %

\subsection{Environments} \label{ssec-env}
We complete our setting with the definition of several environments. Let \textbf{V} be the set of all program variables, \texttt{Loc} be the set of all memory locations, and \texttt{Num} be the set of integers.

As in \cite{HH2010}, we first have a \emph{variable environment} $\gamma$, responsible for mapping a variable name and the block to which it is local ($\lambda$ in the case of global variables) to its bound memory location. This is defined as $\gamma: (\textbf{V} \times \textbf{Bn}) \mapsto \textbf{Loc}$. The notation $\gamma$[\texttt{(X,Bn)} $\Rightarrow$ \texttt{l}] indicates that the pair \texttt{(X,Bn)} maps to the memory location \texttt{l}, while $\gamma$[\texttt{(X,Bn)}] represents an update to $\gamma$ with the mapping for the pair \texttt{(X,Bn)} removed. 

We have a \emph{data store} $\sigma$, responsible for mapping each memory location to the value it currently holds, defined as $\sigma: (\textbf{Loc} \mapsto \textbf{Num})$. The notion $\sigma$[\texttt{l} $\mapsto$ \texttt{v}] indicates location \texttt{l} now holds the value \texttt{v}.

The \emph{procedure environment} $\mu$ is responsible for mapping either a procedure or call identifier to both the actual procedure name (used in code) and (a copy of) the body. This environment is defined as $\mu: (\textbf{Pn} \cup \textbf{Cn}) \mapsto (\textbf{n} \times \textbf{P})$. The notation $\mu$[\texttt{Pn} $\Rightarrow$ \texttt{(n,P)}] represents that \texttt{Pn} maps to the pair \texttt{(n,P)},  $\mu$[\emph{refC(}\texttt{Pn,P}\emph{)}] represents the updating of the mapping for \texttt{Pn} with changes retrieved from \texttt{P}, and $\mu$[\texttt{Pn}] indicates the removal of the mapping for \texttt{Pn} (in each case, \texttt{Pn} could also be \texttt{Cn}).  

Finally, the \emph{while environment} $\beta$ is responsible for mapping a unique loop identifier to a copy of that loop. This serves the purpose of storing both the original condition and program, allowing our semantics to be static. This is defined as $\beta: \textbf{Wn} \mapsto \textbf{P}$. The notation $\beta$[\texttt{Wn} $\Rightarrow$ \texttt{P}] indicates that \texttt{Wn} now maps to the program \texttt{P}, $\beta$[\emph{refW(}\texttt{Wn,P}\emph{)}] represents the updating of the mapping for \texttt{Wn} with changes retrieved from \texttt{P}, and the notation $\beta$[\texttt{Wn}] shows the removal of the mapping for \texttt{Wn}.

We now combine all environments, and use the notation $\square$ to represent the set \{$\sigma$,$\gamma$,$\mu$,$\beta$\}. Each environment has a prime version, indicating a potential and arbitrary change.

\subsection{Scope} \label{ssec-scope}
Local variables can share their name with a global variable, as well as local variables declared in different blocks. The traditional method of handling this, as described for example in \cite{HH2010}, is to implement a stack of environments, storing a copy for each scope. This is not suitable when we use parallel composition as there can be several active scopes in an execution at once. We therefore implement a single environment that will store all versions of variables. Variables will either be global, or local to a specific block. Using $\lambda$ to represent the empty block name (a global variable), associating a variable with the identifier of the block in which it is declared is sufficient. Now all versions of a variable name are stored distinctly. 

We must be able to access the correct version of a given variable name. Under the traditional approach, each environment will have only one mapping of any variable, something we do not have. We must be able to determine the block identifier in which the variable was defined, which will not necessarily be the current block. We achieve this by assigning a \emph{path} to each statement. Each path will be the sequence of the block identifiers \texttt{Bn}, for blocks in which this statement resides, separated using `*'. Consider statement \texttt{F = S (b2*b1,A)} (from Figure \ref{big-ex-re-ann}) that has a path \texttt{b2*b1}, meaning it occurs within a block \texttt{b2}, which is nested within \texttt{b1}. Therefore we have the function \emph{evalV()}, that takes a variable name and a path, traverses the sequence of block names until the first is found that has a local variable of this name, and uses this block name (or $\lambda$ if no match) to return the desired memory location. A similar reasoning, and function \emph{evalP()}, exists to evaluate potentially shared procedure names, returning the correct unique procedure identifier.

One complication is \emph{code reuse}, where the same program code is executed multiple times, with the two cases being procedure and loop bodies. Consider two calls to the same procedure in parallel. Both may create a local variable of a block, where on each side the block has the same name, meaning both will incorrectly use the same version of the local variable. A similar case exists for recursive calls, as shown in Example 1. Therefore, we \emph{rename} any reused code prior to its execution. This must make all constructs unique, a task achieved using the unique call identifier \texttt{Cn}. All construct names are modified to now start with the unique call identifier, with paths also updated to reflect changes made to block identifiers. Consider again the statement for \texttt{F} from Figure \ref{big-ex-re-ann}. When block \texttt{b2} is renamed to \texttt{c1:c2:b2}, the path becomes \texttt{c1:c2:b2*b1} (Figure \ref{big-ex-ann-call}). This removes the issue described above, as each version of the variable will have a different block identifier. This process must modify the call identifier of any recursive call statement similarly. Therefore we have the function \emph{reP()} that implements this renaming of procedure bodies, and as renaming occurs in reverse, we have \emph{IreP()}. 

The reuse of loop bodies is different, as it is not possible for the same code (with same construct names) to be executed in parallel when not in a procedure body (handled above). However, in order to keep all identifiers unique, and to aid future extensions to causal-consistent reversibility, we \emph{version} each construct name, incrementing it by 1 for each loop iteration. For example, a conditional statement \texttt{i1.0} will become \texttt{i1.1}. These versions are maintained via the function \texttt{nextID()}, used in $reL()$ that performs this renaming, and the function \texttt{previousID()}, used in the reverse renaming function $IreL()$.

\section{Forwards Only Operational Semantics} \label{sec-trad}
We now define the traditional, forwards only semantics of our language. We give a set of transition rules for each construct of our language and a set of environments. These rules specify how \emph{configurations} (namely, pairs containing a program and a set of environments) compute by performing single transition steps. The transition rules define our small step transition relation \emph{configuration} $\hookrightarrow$ \emph{configuration}. The transitive closure $\hookrightarrow^*$ represents executions of programs. The forwards only semantics do not perform any state-saving, thus making them irreversible. Transitions labelled with \texttt{a} or \texttt{b} are steps of arithmetic or boolean expression evaluation respectively, while $\hookrightarrow^*_{\texttt{a}}$ and $\hookrightarrow^*_{\texttt{b}}$ are the transitive closure of each. Semantics of both are omitted as they are as expected, see \cite{HH2010}. By abuse of notation, we now use $\square$ to represent all environments of the set \{$\sigma$,$\gamma$,$\mu$,$\beta$\} that are not modified via the specific rule. For example, if a rule only changes the procedure environment $\mu$, then $\square$ will be the set \{$\sigma$,$\gamma$,$\beta$\}. The semantics listed below are static, necessary for later sections including state-saving and our results.

\vspace{.1cm}
\noindent \textbf{Sequential and Parallel Composition} Programs can be of the form \texttt{S;P} or \texttt{P par Q}. As such, programs either execute sequentially, or allow each side of a parallel statement to interleave their execution.

\vspace{-.5cm}
{\small \begin{align*}
&\text{[S1]} \quad \frac{(\texttt{S} \mid \square) \hookrightarrow (\texttt{S$'$} \mid \square')}{(\texttt{S; P} \mid \square) \hookrightarrow (\texttt{S$'$; P} \mid \square')}  &\quad  &\text{[S2]} \quad \frac{}{(\texttt{skip; P} \mid \square) \hookrightarrow (\texttt{P} \mid \square)} \\[6pt]
&\text{[P1]} \quad \frac{(\texttt{P} \mid \square) \hookrightarrow (\texttt{P$'$} \mid \square')}{(\texttt{P par Q} \mid \square) \hookrightarrow (\texttt{P$'$ par Q} \mid \square')}
&\quad &\text{[P2]} \quad \frac{(\texttt{Q} \mid \square) \hookrightarrow (\texttt{Q$'$} \mid \square')}{(\texttt{P par Q} \mid \square) \hookrightarrow (\texttt{P par Q$'$} \mid \square')} \\[6pt]
&\text{[P3]} \quad \frac{}{(\texttt{P par skip} \mid \square) \hookrightarrow (\texttt{P} \mid \square)} &\quad
&\text{[P4]} \quad \frac{}{(\texttt{skip par Q} \mid \square) \hookrightarrow (\texttt{Q} \mid \square)}
\end{align*} }%

\vspace{.1cm}
\noindent \textbf{Assignment} All assignments are considered destructive, with the overwritten value being lost. A single atomic rule both evaluates the expression and assigns the new value to the appropriate memory location. Arithmetic expressions do not contain side effects, meaning evaluation of these does not change the environments, as shown in rule \text{[D1]}. Similarly, this is also the case for boolean expressions, as shown in rule \text{[I1]} and there after. 

\vspace{-.5cm}
{\small \begin{align*}
&\text{[D1]} \quad \frac{(\texttt{e pa} \mid \sigma,\gamma,\square) \hookrightarrow^*_{\texttt{a}} (\texttt{v} \mid \sigma,\gamma,\square) \quad \texttt{$evalV(\gamma$,pa,X$)$ = l}}{(\texttt{X = e pa} \mid \sigma,\gamma,\square) \hookrightarrow (\texttt{skip} \mid \sigma[\texttt{l $\mapsto$ v}],\gamma,\square)}
\end{align*} }%

\vspace{.1cm}
\noindent \textbf{Conditional} Condition evaluation is atomic via \text{[I1]}, before the appropriate branch is executed completely to skip (potentially interleaved).

\vspace{-.5cm}
{\small \begin{align*}
&\text{[I1]} \quad \frac{(\texttt{b pa} \mid \square) \hookrightarrow^*_{\texttt{b}} (\texttt{V} \mid \square)}{(\texttt{if In b then P else Q end pa} \mid \square) \hookrightarrow (\texttt{if In V then P else Q end pa} \mid \square)} \\[6pt]
&\text{[I2]} \quad \frac{(\texttt{P} \mid \square) \hookrightarrow (\texttt{P$'$} \mid \square')}{(\texttt{if In T then P else Q end pa} \mid \square) \hookrightarrow (\texttt{if In T then P$'$ else Q end pa} \mid \square')} \\[6pt]
&\text{[I3]} \quad \frac{(\texttt{Q} \mid \square) \hookrightarrow (\texttt{Q$'$} \mid \square')}{(\texttt{if In F then P else Q end pa} \mid \square) \hookrightarrow (\texttt{if In F then P else Q$'$ end pa} \mid \square')} \\[6pt]
&\text{[I4]} \quad \frac{}{(\texttt{if In T then skip else Q end pa} \mid \square) \hookrightarrow (\texttt{skip} \mid \square)}  \\[6pt]
&\text{[I5]} \quad \frac{}{(\texttt{if In F then P else skip end pa} \mid \square) \hookrightarrow (\texttt{skip} \mid \square)} 
\end{align*} }%

\vspace{.1cm}
\noindent \textbf{While Loop} Evaluation of the condition is always atomic. \text{[W1]} handles the first iteration of a loop, where no mapping for \texttt{Wn} is present in $\beta$. The mapping \texttt{Wn} $\Rightarrow$ \texttt{R} is inserted, and the condition is evaluated. \text{[W2]} handles any other iteration, evaluating the condition and updating the mapping within $\beta$ to \texttt{Wn}~$\Rightarrow$~\texttt{R$'$}. Both rules rename the body, making constructs unique. \text{[W3]} executes the body. \text{[W4]} continues the loop using the program \texttt{P}, retrieved from $\beta$ for \texttt{Wn}, until the condition is false, when \text{[W5]} will conclude the statement. The premise \texttt{$\beta($Wn$)$ = R} indicates an arbitrary mapping exists. This is necessary as the rule \text{[W5]} requires the removal of some mapping via $\beta[\texttt{Wn}]$. Note that these semantics (and the semantics defined in Sections \ref{sec-ann} and \ref{sec-inv}) are correct for all while loops with conditions that require evaluation. In the case of \texttt{b} initially being \texttt{T} or \texttt{F}, there may be ambiguity in our rules.

\vspace{-.5cm}
{\small \begin{align*}
&\text{[W1]} \quad \frac{\texttt{$\beta($Wn$)$ = }\emph{und} \quad (\texttt{b pa} \mid \beta,\square) \hookrightarrow^*_{\texttt{b}} (\texttt{V} \mid \beta,\square) }{(\texttt{while Wn b do P end pa} \mid \beta,\square) \hookrightarrow (\texttt{while Wn V do $reL($P$)$ end pa} \mid \beta[\texttt{Wn} \Rightarrow \texttt{R}],\square)} \\[2pt] & \phantom{\text{[W1]} \quad } \text{where } \texttt{R} = \texttt{while Wn b do $reL($P$)$ end pa} \\[6pt]
&\text{[W2]} \quad \frac{\texttt{$\beta($Wn$)$ = while Wn b do P end pa} \quad (\texttt{b pa} \mid \beta,\square) \hookrightarrow^*_{\texttt{b}} (\texttt{V} \mid \beta,\square) }{(\texttt{while Wn b do P end pa} \mid \beta,\square) \hookrightarrow (\texttt{while Wn V do $reL($P$)$ end pa} \mid \beta[\texttt{Wn} \Rightarrow \texttt{R$'$}],\square)} \\[2pt] & \phantom{\text{[W2]} \quad } \text{where }  \texttt{R$'$} = \texttt{while Wn b do $reL($P$)$ end pa}  \\[6pt]
&\text{[W3]} \quad \frac{(\texttt{R} \mid \square) \hookrightarrow (\texttt{R$'$} \mid \square')}{(\texttt{while Wn T do R end pa} \mid \square) \hookrightarrow (\texttt{while Wn T do R$'$ end pa} \mid \square')}  \\[6pt]
&\text{[W4]} \quad \frac{\texttt{$\beta($Wn$)$ = P} }{(\texttt{while Wn T do skip end pa} \mid \beta,\square) \hookrightarrow (\texttt{P} \mid \beta,\square)}  \\[6pt]
&\text{[W5]} \quad \frac{\texttt{$\beta($Wn$)$ = R}}{(\texttt{while Wn F do P end pa} \mid \beta,\square) \hookrightarrow (\texttt{skip} \mid \beta[\texttt{Wn}],\square)}  
\end{align*} }%

\vspace{.1cm}
\noindent \textbf{Block} Blocks begin with \text{[B1]} that creates the \texttt{runB} construct. This then executes the block body via \text{[B2]}, beginning with the declaration of local variables and procedures. The program then executes using these local definitions, before all such information is removed. Finally \text{[B3]} concludes the statement.

\vspace{-.5cm}
{\small \begin{align*}
&\text{[B1]} \quad \frac{}{(\texttt{begin Bn P end} \mid \square) \hookrightarrow (\texttt{runB P end} \mid \square)} \quad \text{where } \texttt{P } \text{=} \texttt{ DV;DP;Q;RP;RV} \\[6pt] 
&\text{[B2]} \quad \frac{(\texttt{P} \mid \square) \hookrightarrow (\texttt{P$'$} \mid \square')}{(\texttt{runB P end} \mid \square) \hookrightarrow (\texttt{runB P$'$ end} \mid \square')}  \quad \text{[B3]} \quad \frac{}{(\texttt{runB skip end} \mid \square) \hookrightarrow (\texttt{skip} \mid \square)} 
\end{align*} }%

\vspace{.1cm}
\noindent \textbf{Variable and Procedure Declaration} A variable declaration \text{[L1]} associates the given variable name and current block name \texttt{Bn} (first element of the sequence \texttt{pa}, written \texttt{Bn*pa$'$}) to the next available memory location \texttt{l} (via \texttt{nextLoc()}) in $\gamma$, while also mapping this location to the value \texttt{v} in $\sigma$ (via the notation \texttt{l} $\mapsto$ \texttt{v}). A procedure declaration \text{[L2]} inserts the basis mapping between the unique procedure identifier \texttt{Pn} and a pair containing both the procedure name and the procedure body \texttt{(n,P)}. A call statement uses this mapping to create a renamed version. 

\vspace{-.5cm}
{\small \begin{align*}
&\text{[L1]} \quad \frac{\texttt{nextLoc() = l} \quad \texttt{pa} = \texttt{Bn*pa$'$}}{(\texttt{var X = v pa} \mid \sigma,\gamma,\square) \hookrightarrow (\texttt{skip} \mid \sigma[\texttt{l} \mapsto \texttt{v}],\gamma[(\texttt{X},\texttt{Bn}) \Rightarrow \texttt{l}],\square)} \\[6pt]
&\text{[L2]} \quad \frac{}{(\texttt{proc Pn n is P pa} \mid \mu,\square) \hookrightarrow (\texttt{skip} \mid \mu[\texttt{Pn} \Rightarrow \texttt{(n,P)}],\square)}
\end{align*} }%

\vspace{.1cm}
\noindent \textbf{Procedure Call} \text{[G1]} evaluates the procedure name to \texttt{Pn}, and retrieves the basis entry \texttt{(n,P)} from $\mu$. The renamed version of \texttt{P}, written \texttt{P$'$}, is then inserted into $\mu$ via the mapping \texttt{Cn} $\Rightarrow$ \texttt{(n,P$'$)}, and the \texttt{runC} construct is formed. \text{[G2]} executes the body of the call statement, before \text{[G3]} concludes the statement by removing the mapping for \texttt{Cn} within $\mu$, written $\mu$[\texttt{Cn}]. 

\vspace{-.3cm}
{\small \begin{align*}
&\text{[G1]} \quad \frac{\texttt{$evalP($n,pa$)$ = Pn} \quad \texttt{$\mu($Pn$)$ = (n,P)} \quad \texttt{$reP($P,Cn$)$ = P$'$}}{(\texttt{call Cn n pa} \mid \mu,\square) \hookrightarrow (\texttt{runC Cn P$'$ end} \mid \mu[\texttt{Cn} \Rightarrow \texttt{(n,P$'$)}],\square)} \\[6pt]
&\text{[G2]} \quad \frac{(\texttt{P} \mid \square) \hookrightarrow (\texttt{P$'$} \mid \square')}{(\texttt{runC Cn P end} \mid \square) \hookrightarrow (\texttt{runC Cn P$'$ end} \mid \square')} \\[6pt]
&\text{[G3]} \quad \frac{}{(\texttt{runC Cn skip end} \mid \mu,\square) \hookrightarrow (\texttt{skip} \mid \mu[\texttt{Cn}],\square)}  
\end{align*} }%

\vspace{.1cm}
\noindent \textbf{Variable and Procedure Removal} A local variable \texttt{X} is local to the inner most block \texttt{Bn} of its path, written as \texttt{Bn*pa$'$}. Removal of a variable \text{[H1]} removes the mapping of \texttt{(X,Bn)} from $\gamma$, written $\gamma$[\texttt{(X,Bn)}]. The location associated with this mapping is set to \texttt{0} (\texttt{l} $\mapsto$ \texttt{0}) and marked free for future use. A procedure removal \text{[H2]} removes the  mapping for the given procedure identifier \texttt{Pn}, written as $\mu$[\texttt{Pn}].

\vspace{-.3cm}
{\small \begin{align*}
&\text{[H1]} \quad \frac{\texttt{pa} = \texttt{Bn*pa$'$} \quad \texttt{$\gamma($X,Bn$)$ = l}}{(\texttt{remove X = v pa} \mid \sigma,\gamma,\square) \hookrightarrow (\texttt{skip} \mid \sigma[\texttt{l} \mapsto \texttt{0}],\gamma[(\texttt{X},\texttt{Bn})],\square)} \\[6pt]
&\text{[H2]} \quad \frac{}{(\texttt{remove Pn n is P pa} \mid \mu,\square) \hookrightarrow (\texttt{skip} \mid \mu[\texttt{Pn}],\square)}
\end{align*} }%

\section{Forwards Semantics of Annotated Programs} \label{sec-ann}
Annotation is the process of generating our annotated version of a program. This will make small changes to the syntax of our program, adding the capability of storing necessary reversal information. Before defining this in detail, we must have an environment for storing this information, keeping it separate from the program state. To do this, we use an updated version of our auxiliary store $\delta$ \cite{JH2017}. There is a stack for each program variable name, storing any overwritten values the variable holds throughout the execution, whether the variable is global, local or both. Using one stack for all versions of a variable helps us to devise a technique to handle \emph{races} on that variable. There is a single stack \texttt{B} for \emph{all} conditional statements. After completion of the conditional, a pair containing an identifier and a boolean value indicating which branch was executed will be saved. Identifiers allow us to resolve any races, meaning all pairs for all conditionals can be pushed to a single stack. There is a single stack \texttt{W} for \emph{all} while loops. As explained and illustrated in our previous work \cite{JH2017}, this stack will contain pairs of an identifier and a boolean value. These pairs produce a sequence of boolean values necessary for inverse execution. Stack \texttt{WI} stores all \emph{annotation information} (order of identifiers) of a while loop, before it's removed from $\beta$. Stack \texttt{Pr} performs a similar task for procedure bodies. Let \textbf{V} be the set of all variables, $\mathbb{S}(\textbf{V})$ be a set of stacks, one for each element of $\textbf{V}$, $\mathbb{B}$ be the set of boolean values, $\mathbb{C}$ be the annotation information and \textbf{K} be the set of identifiers. Then $ \delta : (\mathbb{S}(\textbf{V}) \mapsto (\textbf{K} \times \textbf{Num})) \cup (\mathbb{S}(\texttt{B}) \mapsto (\textbf{K} \times \mathbb{B})) \cup (\mathbb{S}(\texttt{W}) \mapsto (\textbf{K} \times \mathbb{B})) \cup (\mathbb{S}(\texttt{WI}) \mapsto (\textbf{K} \times \mathbb{C})) \cup (\mathbb{S}(\texttt{Pr}) \mapsto (\textbf{K} \times \mathbb{C})) $. The notation $\delta$[\texttt{el} $\rightharpoonup$ \texttt{st}] pushes \texttt{el} to the stack \texttt{st}, while $\delta$[\texttt{st/st$'$}] pops the top element of stack \texttt{st}, leaving the remaining stack \texttt{st$'$}.

We now define annotation. Each statement, excluding blocks and parallel, receives a stack \texttt{A} for identifiers. The association of a statement to its stack persists throughout the execution. Each time a statement executes, the (global and atomic) function \texttt{next()} retrieves the next available identifier, which is pushed to that statements stack (and $\delta$ if necessary). The functions $ann()$ and $a()$ are defined below. 

\vspace{-.2cm}
{\small $$ ann(\texttt{$\varepsilon$}) = \texttt{$\varepsilon$} \quad ann(\texttt{S;P}) = a(\texttt{S}); ann(\texttt{P}) \quad ann(\texttt{P par Q}) = ann(\texttt{P}) \texttt{ par } ann(\texttt{Q}) $$ }
\vspace{-.9cm}
{\small \begin{align*}
a(\texttt{skip}) &= \texttt{skip I} \\
a(\texttt{X = e pa}) &= \texttt{X = e (pa,A)} \\
a(\texttt{if In b then P else Q end pa}) &= \texttt{if In b then $ann($P$)$ else $ann($Q$)$ end (pa,A)} \\
a(\texttt{while Wn b do P end pa}) &= \texttt{while Wn b do $ann($P$)$ end (pa,A)} \\
a(\texttt{begin Bn DV DP P RP RV end}) &= \texttt{begin Bn $ann($DV$)$ $ann($DP$)$ $ann($P$)$ $ann($RP$)$ $ann($RV$)$ end} \\
a(\texttt{var X = v pa}) &= \texttt{var X = v (pa,A)} \\
a(\texttt{proc Pn n is P pa}) &= \texttt{proc Pn n is $ann($P$)$ (pa,A)} \\
a(\texttt{call Cn n pa}) &= \texttt{call Cn n (pa,A)} \\
a(\texttt{remove X = v pa}) &= \texttt{remove X = v (pa,A)} \\
a(\texttt{remove Pn n is P pa}) &= \texttt{remove Pn n is $ann($P$)$ (pa,A)} 
\end{align*} }%
We use \texttt{I} to represent either nothing, a path, an identifier stack, or a pair consisting of both a path and an identifier stack. After application of these functions, the resulting annotated version is of the following, modified syntax (with expressions omitted as they match Section \ref{sec-les}). We note that $a(\texttt{runB P end})$ is $\texttt{runB $ann($P$)$ end}$, and $a(\texttt{runC Cn P end})$ is $\texttt{runC Cn $ann($P$)$ end A}$. Execution of the annotated version produces the \emph{executed annotated version}, an identical copy but with populated identifier stacks. 
\begin{align*}
\texttt{AP} &::= \texttt{$\varepsilon$} ~|~ \texttt{AS} ~|~ \texttt{AP; AP} ~|~ \texttt{AP par AP} \\
\texttt{AS} &::= \texttt{skip I} ~|~  \texttt{X = E (pa,A)} ~|~ \texttt{if In B then AP else AQ end (pa,A)} ~| \\ &\phantom{\texttt{::=}}  \texttt{while Wn B do AP end (pa,A)} ~|~ \texttt{begin Bn ADV ADP AP ARP ARV end}  ~| \\ &\phantom{\texttt{::=}}\texttt{call Cn n (pa,A)} ~|~   \texttt{runC Cn AP end A} ~|~   \texttt{runB AP end} 
\end{align*} 
\vspace{-0.9cm}
\begin{align*}
\texttt{ADV} &::= \texttt{$\varepsilon$} ~|~ \texttt{var X = v (pa,A); ADV} &
\texttt{ADP} &::= \texttt{$\varepsilon$} ~|~ \texttt{proc Pn n is AP (pa,A); ADP} \\
\texttt{ARV} &::= \texttt{$\varepsilon$} ~|~ \texttt{remove X = v (pa,A); ARV} &
\texttt{ARP} &::= \texttt{$\varepsilon$} ~|~ \texttt{remove Pn n is AP (pa,A); ARP} 
\end{align*}%

As all statements within an annotated version will be of this syntax, this must be reflected in our environments, specifically $\mu$ and $\beta$ that store programs. As a result, we now use $\square$ to represent the set of annotated environments, unless explicitly stated otherwise. 

\begin{figure}[t]
  \begin{minipage}[b]{0.49\linewidth}
   \centering
   {\small \begin{lstlisting}[xleftmargin=2.0ex,mathescape=true]
$\texttt{begin b1}$  
 $\texttt{proc p1 fib is}$ 
  $\texttt{begin b2}$
   $\texttt{var T}$ = $\texttt{0 b2}$;
   $\texttt{if i1}$ $\texttt{(N - 2 > 0)}$ $\texttt{then}$
    $\texttt{T}$ = $\texttt{F + S b2}$;
    $\texttt{F}$ = $\texttt{S b2}$;
    $\texttt{S}$ = $\texttt{T b2}$;
    $\texttt{N}$ = $\texttt{N - 1 b2}$;
    $\texttt{call c2 fib b2}$;
   $\texttt{end  b2}$
   $\texttt{remove T}$ = $\texttt{0 b2}$;
  $\texttt{end}$
 $\texttt{end  b1}$
 $\texttt{call c1 fib b1}$;
 $\texttt{remove p1 fib is P  b1}$;  
$\texttt{end}$
\end{lstlisting} }
    \caption{Original program}
	\label{big-ex-org}
  \end{minipage}
  \hspace{0.01cm}
  \begin{minipage}[b]{0.49\linewidth}
    \centering
    {\small \begin{lstlisting}[xleftmargin=2.0ex,mathescape=true]
$\texttt{begin b1}$    
 $\texttt{proc p1 fib is}$ 
  $\texttt{begin b2}$
   $\texttt{var T}$ = $\texttt{0 (b2*b1,A)}$;
   $\texttt{if i1}$ $\texttt{(N - 2 > 0)}$ $\texttt{then}$
    $\texttt{T}$ = $\texttt{F + S (b2*b1,A)}$;
    $\texttt{F}$ = $\texttt{S (b2*b1,A)}$;
    $\texttt{S}$ = $\texttt{T (b2*b1,A)}$;
    $\texttt{N}$ = $\texttt{N - 1 (b2*b1,A)}$;
    $\texttt{call c2 fib (b2*b1,A)}$;
   $\texttt{end  (b2*b1,A)}$
   $\texttt{remove T}$ = $\texttt{0 (b2,A)}$;
  $\texttt{end}$
 $\texttt{end  (b1,A)}$   
 $\texttt{call c1 fib (b1,A)}$; 
 $\texttt{remove p1 fib is AP (b1,A)}$;  
$\texttt{end}$
	\end{lstlisting} }
    \caption{Renamed and annotated program}
    \label{big-ex-re-ann}
  \end{minipage}
\end{figure}

\begin{example} We now consider the example implementation of a Fibonacci sequence using our programming language above, shown in Figure \ref{big-ex-org}. Note that version numbers are omitted due to the absence of while loops, and all paths are initially just the most direct block name. Let \texttt{P} denote the procedure body shown in the declaration statement, namely lines 3-13 of Figure \ref{big-ex-org}. The procedure removal statement on line 16 then uses this. Assume global variables \texttt{F=3}, \texttt{S=4} and \texttt{N=4}. This program calculates the \texttt{N}th element of the Fibonacci sequence beginning with the first and second elements \texttt{F} and \texttt{S}. After execution, the \texttt{N}th element will be the value held by \texttt{S}. Before we execute this program forwards, it must first be both annotated and renamed, shown in Figure \ref{big-ex-re-ann}. All paths have now been updated to include all block identifiers necessary for execution, and all appropriate statements now have a stack \texttt{A} for storing identifiers. Line 7 of Figure \ref{big-ex-org} has become line 7 of Figure \ref{big-ex-re-ann}, which now has the path \texttt{b2*b1}, meaning this statement appears directly within \texttt{b2}, and indirectly within \texttt{b1}.  \end{example}\label{ex-1}

Prior to defining the operational semantics of this, we must first introduce three functions. The first, $getAI()$, returns the order and application of identifiers to a given program. This allows us to extract the annotation information that would otherwise be lost when a while or procedure environment mapping is removed. The second, $\emph{refW()}$, reflects a given annotation change to the copy of the program mapped to the given while identifier. The third, $\emph{refC()}$, is identical, but will reflect a change made to a procedure body using a given call identifier. Recall functions $reL()$ and $reP()$ from Section \ref{ssec-scope}.

We are now ready to give the operational semantics. The transition rules are those in Section \ref{sec-trad}, but with $\hookrightarrow$ replaced with $\rightarrow$, and with all state-saving performed. We introduce \emph{m-rules}, the name given to each transition rule that assigns an identifier. All other rules that do not use identifiers are now named \emph{non m-rules}. Transitions $\tran{m}$ are also called \emph{identifier transitions}.

\vspace{.1cm}
\noindent \textbf{Sequential and Parallel Composition} These are identical to Section \ref{sec-trad}, but with the annotated syntax.

\vspace{-.5cm}
{\small \begin{align*}
&\text{[S1a]} \quad \frac{(\texttt{AS} \mid \square) \tran{\circ} (\texttt{AS$'$} \mid \square')}{(\texttt{AS; AP} \mid \square) \tran{\circ} (\texttt{AS$'$; AP} \mid \square')} & \quad &\text{[S2a]} \quad \frac{}{(\texttt{skip I; AP} \mid \square) \rightarrow (\texttt{AP} \mid \square)} \\[6pt]
&\text{[P1a]} \quad \frac{(\texttt{AP} \mid \square) \tran{\circ} (\texttt{AP$'$} \mid \square')}{(\texttt{AP par AQ} \mid \square) \tran{\circ} (\texttt{AP$'$ par AQ} \mid \square')} &\quad &\text{[P2a]} \quad \frac{(\texttt{AQ} \mid \square) \tran{\circ} (\texttt{AQ$'$} \mid \square')}{(\texttt{AP par AQ} \mid \square) \tran{\circ} (\texttt{AP par AQ$'$} \mid \square')} \\[6pt]
&\text{[P3a]} \quad \frac{}{(\texttt{AP par skip I} \mid \square) \rightarrow (\texttt{AP} \mid \square)} &\quad
&\text{[P4a]} \quad \frac{}{(\texttt{skip I par AQ} \mid \square) \rightarrow (\texttt{AQ} \mid \square)} 
\end{align*} }%

\vspace{.1cm}
\noindent \textbf{Assignment} This is an m-rule, saving the old value and the next available identifier \texttt{m}, retrieved via the function \texttt{next()}, onto this variables stack on $\delta$.

\vspace{-.5cm}
{\small \begin{align*}
&\text{[D1a]} \quad \frac{(\texttt{e pa} \mid \delta,\sigma,\gamma,\square) \hookrightarrow^*_{\texttt{a}} (\texttt{v} \mid \delta,\sigma,\gamma,\square) \quad \texttt{m = next()} \quad \texttt{$evalV(\gamma$,pa,X$)$ = l}}{(\texttt{X = e (pa,A)} \mid \delta,\sigma,\gamma,\square) \tran{m} (\texttt{skip m:A} \mid \delta[\texttt{(m,$\sigma($l$)$) $\rightharpoonup$ X}],\sigma[\texttt{l} \mapsto \texttt{v}],\gamma,\square)}
\end{align*} }%

\vspace{.1cm}
\noindent \textbf{Conditional} All rules follow as in Section \ref{sec-trad}, but with annotated programs, and \text{[I4a]} and \text{[I5a]} both being m-rules. These save the next available identifier \texttt{m} (via \texttt{next()}) and a boolean value indicating which branch was executed (after execution of the branch) onto stack \texttt{B} on $\delta$.

\vspace{-.5cm}
{\small \begin{align*}
&\text{[I1a]} \quad \frac{(\texttt{b pa} \mid \square) \hookrightarrow^*_{\texttt{b}} (\texttt{V} \mid \square)}{(\texttt{if In b then AP else AQ end (pa,A)} \mid \square) \rightarrow (\texttt{if In V then AP else AQ end (pa,A)} \mid \square)} \\[6pt]
&\text{[I2a]} \quad \frac{(\texttt{AP} \mid \square) \tran{\circ} (\texttt{AP$'$} \mid \square')}{(\texttt{if In T then AP else AQ end (pa,A)} \mid \square) \tran{\circ} (\texttt{if In T then AP$'$ else AQ end (pa,A)} \mid \square')} \\[6pt]
&\text{[I3a]} \quad \frac{(\texttt{AQ} \mid \square) \tran{\circ} (\texttt{AQ$'$} \mid \square')}{(\texttt{if In F then AP else AQ end (pa,A)} \mid \square) \tran{\circ} (\texttt{if In F then AP else AQ$'$ end (pa,A)} \mid \square')} \\[6pt]
&\text{[I4a]} \quad \frac{\texttt{m = next()}}{(\texttt{if In T then skip I else AQ end (pa,A)} \mid \delta,\square) \tran{m} (\texttt{skip m:A} \mid \delta[\texttt{(m,T) $\rightharpoonup$ B}],\square)} \\[6pt]
&\text{[I5a]} \quad \frac{\texttt{m = next()}}{(\texttt{if In F then AP else skip I end (pa,A)} \mid \delta,\square) \tran{m} (\texttt{skip m:A} \mid \delta[\texttt{(m,F) $\rightharpoonup$ B}],\square)} 
\end{align*} }%

\vspace{.1cm}
\noindent \textbf{While Loop} The first two rules are m-rules, saving the next available identifier \texttt{m} (via \texttt{next()}) and an element of the boolean sequence onto stack \texttt{W} on $\delta$. \text{[W1a]} handles the first iteration of a loop, creating a mapping on $\beta$ as in Section \ref{sec-trad} (but with annotated programs), and saving the first element \texttt{F} of the boolean sequence (see \cite{JH2017}). \text{[W2a]} handles any other iteration, updating the current mapping as before and saving the next element \texttt{T} of the boolean sequence (see \cite{JH2017}). Both rules rename the loop body. The body executes via \text{[W3a]}, now reflecting all annotation changes of \texttt{AR$'$} into the stored copy, written using $\beta'$[\emph{refW(}\texttt{Wn,AR$'$}\emph{)}]. Finally, a loop either continues via \text{[W4a]}, or finishes via the m-rule \text{[W5a]}. This final rule stores the next available identifier \texttt{m} and all annotation information ($getAI()$) onto stack \texttt{WI}, before removing the mapping, written $\beta$[\texttt{Wn}]. 

\vspace{-.4cm}
{\small \begin{align*}
&\text{[W1a]} \quad \frac{\texttt{m = next()} \quad \texttt{$\beta($Wn$)$ = }\emph{und} \quad (\texttt{b pa} \mid \beta,\square) \hookrightarrow^*_{\texttt{b}} (\texttt{V} \mid \beta,\square) }{(\texttt{S} \mid \delta,\beta,\square) \tran{m} (\texttt{while Wn V do $reL($AP$)$ end (pa,m:A)} \mid \delta[\texttt{(m,F) $\rightharpoonup$ W}],\beta[\texttt{Wn} \Rightarrow \texttt{AR}],\square)} \\[5pt] & \phantom{\text{[W1a]} \quad } \text{where } \texttt{S} = \texttt{while Wn b do AP end (pa,A)} \text{ and } \texttt{AR} = \texttt{while Wn b do $reL($AP$)$ end (pa,m:A)}  \\[6pt]
&\text{[W2a]} \quad \frac{\texttt{m = next()} \quad \texttt{$\beta($Wn$)$ = while Wn b do AQ end (pa,A)} \quad (\texttt{b pa} \mid \beta,\square) \hookrightarrow^*_{\texttt{b}} (\texttt{V} \mid \beta,\square) }{(\texttt{S} \mid \delta,\beta,\square) \tran{m} (\texttt{while Wn V do $reL($AQ$)$ end (pa,m:A)} \mid \delta[\texttt{(m,T) $\rightharpoonup$ W}],\beta[\texttt{Wn} \Rightarrow \texttt{AR$'$}],\square)} \\[5pt] & \phantom{\text{[W2a]} \quad } \text{where } \texttt{S} = \texttt{while Wn b do AP end (pa,A)} \text{ and } \texttt{AR$'$} = \texttt{while Wn b do $reL($AQ$)$ end (pa,m:A)}  \\[6pt]
&\text{[W3a]} \quad \frac{(\texttt{AR} \mid \beta,\square) \tran{\circ} (\texttt{AR$'$} \mid \beta',\square')}{(\texttt{while Wn T do AR end (pa,A)} \mid \beta,\square) \tran{\circ} (\texttt{while Wn T do AR$'$ end (pa,A)} \mid \beta'',\square')} \\[4pt] & \phantom{\text{[W3a]} \quad } \text{where } \texttt{$\beta''$} = \texttt{$\beta'$}[\emph{refW(}\texttt{Wn,AR$'$}\emph{)}] \\[6pt]
&\text{[W4a]} \quad \frac{\texttt{$\beta($Wn$)$ = AP} }{(\texttt{while Wn T do skip I end (pa,A)} \mid \beta,\square) \rightarrow (\texttt{AP} \mid \beta,\square)}  \\[6pt]
&\text{[W5a]} \quad \frac{\texttt{m = next()} \quad \texttt{$\beta($Wn$)$ = AR}}{(\texttt{while Wn F do AP end (pa,A)} \mid \delta,\beta,\square) \tran{m} (\texttt{skip m:A} \mid \delta[\texttt{(m,$getAI(\beta($Wn$))$) $\rightharpoonup$ WI}],\beta[\texttt{Wn}],\square)}  
\end{align*} }%

\vspace{.1cm}
\noindent \textbf{Block} These are identical to before, but using the annotated syntax.

\vspace{-.5cm}
{\small \begin{align*}
&\text{[B1a]} \quad \frac{}{(\texttt{begin Bn AP end} \mid \square) \rightarrow (\texttt{runB AP end} \mid \square)} \quad \text{where } \texttt{AP } \text{=} \texttt{ ADV;ADP;AQ;ARP;ARV} \\[6pt] 
&\text{[B2a]} \quad \frac{(\texttt{AP} \mid \square) \tran{\circ} (\texttt{AP$'$} \mid \square')}{(\texttt{runB AP end} \mid \square) \tran{\circ} (\texttt{runB AP$'$ end} \mid \square')}  \quad \text{[B3a]} \quad \frac{}{(\texttt{runB skip I end} \mid \square) \rightarrow (\texttt{skip} \mid \square)}  
\end{align*} }%

\vspace{.1cm}
\noindent \textbf{Variable and Procedure Declaration} Variable declarations \text{[L1a]} are as before, but are now m-rules without state-saving. Procedure declarations \text{[L2a]} are also as before, but is also an m-rule without state-saving, and all programs are now annotated.

\vspace{-.5cm}
{\small \begin{align*}
&\text{[L1a]} \quad \frac{\texttt{m = next()} \quad \texttt{nextLoc() = l} \quad \texttt{pa} = \texttt{Bn*pa$'$} }{(\texttt{var X = v (pa,A)} \mid \sigma,\gamma,\square) \tran{m} (\texttt{skip m:A} \mid \sigma[\texttt{l} \mapsto \texttt{v}],\gamma[(\texttt{X},\texttt{Bn}) \Rightarrow \texttt{l}],\square)} \\[6pt]
&\text{[L2a]} \quad \frac{\texttt{m = next()}}{(\texttt{proc Pn n is AP (pa,A)} \mid \mu,\square) \tran{m} (\texttt{skip m:A} \mid \mu[\texttt{Pn} \Rightarrow \texttt{(n,AP)}],\square)}
\end{align*} }%

\vspace{.1cm}
\noindent \textbf{Procedure Call} \text{[G1a]} inserts a renamed copy of the basis entry onto $\mu$, exactly as in Section \ref{sec-trad}, but with an annotated program. \text{[G2a]} uses the \texttt{runC} construct to execute this renamed copy, but now reflects any annotation changes to the stored copy, written $\mu'$[\emph{refC(}\texttt{Cn},\texttt{AP$'$}\emph{)}]. \text{[G3a]} is now an m-rule, saving all annotation changes from the copy to the stack \texttt{Pr}, alongside the next available identifier \texttt{m} (via \texttt{next()}). The renamed copy is removed from $\mu$, written as $\mu$[\texttt{Cn}]. 

\vspace{-.5cm}
{\small \begin{align*}
&\text{[G1a]} \quad \frac{\texttt{$evalP($n,pa$)$ = Pn} \quad \texttt{$\mu($Pn$)$ = (n,AP)} \quad \texttt{$reP($AP,Cn$)$ = AP$'$}}{(\texttt{call Cn n (pa,A)} \mid \mu,\square) \rightarrow (\texttt{runC Cn AP$'$ end A} \mid \mu[\texttt{Cn} \Rightarrow \texttt{(n,AP$'$)}],\square)} \\[6pt]
&\text{[G2a]} \quad \frac{(\texttt{AP} \mid \mu,\square) \tran{\circ} (\texttt{AP$'$} \mid \mu',\square')}{(\texttt{runC Cn AP end A} \mid \mu,\square) \tran{\circ} (\texttt{runC Cn AP$'$ end A} \mid \mu'[\emph{refC(}\texttt{Cn},\texttt{AP$'$}\emph{)}],\square')} \\[6pt]
&\text{[G3a]} \quad \frac{\texttt{m = next()} \quad \texttt{$\mu$(Cn) = AP}}{(\texttt{runC Cn skip I end A} \mid \delta,\mu,\square) \tran{m} (\texttt{skip m:A} \mid \delta[\texttt{(m,$getAI(\mu($Cn$))$) $\rightharpoonup$ Pr}],\mu[\texttt{Cn}],\square)} 
\end{align*} }%

\vspace{.1cm}
\noindent \textbf{Variable and Procedure Removal} Variable removal \text{[H1a]} is now an m-rule, saving the final value of that variable and the next available identifier \texttt{m} (via \texttt{next()}) onto that variables stack on $\delta$. The mapping is removed as in Section \ref{sec-trad}. Procedure removal \text{[H2a]} is as before, but an m-rule without state-saving.

\vspace{-.5cm}
{\small \begin{align*}
&\text{[H1a]} \quad \frac{\texttt{m = next()} \quad \texttt{pa} = \texttt{Bn*pa$'$} \quad \texttt{$\gamma($X,Bn$)$ = l}}{(\texttt{remove X = v (pa,A)} \mid \delta,\sigma,\gamma,\square) \tran{m} (\texttt{skip m:A} \mid \delta[\texttt{(m,$\sigma($l$)$) $\rightharpoonup$ X}],\sigma[\texttt{l} \mapsto \texttt{0}],\gamma[(\texttt{X},\texttt{Bn})],\square)} \\[6pt]
&\text{[H2a]} \quad \frac{\texttt{m = next()}  \quad \texttt{$\mu$(Pn) = AQ}}{(\texttt{remove Pn n is AP (pa,A)} \mid \mu,\square) \tran{m} (\texttt{skip m:A} \mid \mu[\texttt{Pn}],\square)}
\end{align*} }%

\subsection{Results} \label{ssec-ann-res}
We first define equivalence between traditional and annotated environments. 

\begin{defn}{} \label{def-sig-gam}
Let $\sigma$ be a data store, $\sigma_1$ be an annotated data store, $\gamma$ be a variable environment and $\gamma_1$ be an annotated variable environment. We have $(\sigma,\gamma)$ is \emph{equivalent} to $(\sigma_1,\gamma_1)$, written $(\sigma,\gamma) \approx_S (\sigma_1,\gamma_1)$, if and only if \emph{dom(}$\gamma$\emph{)} = \emph{dom(}$\gamma_1$\emph{)} and $\sigma(\gamma(\textup{\texttt{X}},\textup{\texttt{Bn}}))$ = $\sigma_1(\gamma_1(\textup{\texttt{X}},\textup{\texttt{Bn}}))$ for all \texttt{X} $\in$ \emph{dom(}$\gamma$\emph{)} and block names \texttt{Bn}.
\end{defn}

\begin{defn}{} \label{def-proc}
Let $\mu$ be a procedure environment, and $\mu_1$ be an annotated procedure environment. We have that $\mu$ is \emph{equivalent} to $\mu_1$, written $\mu \approx_P \mu_1$, if and only if \emph{dom(}$\mu$\emph{)} = \emph{dom(}$\mu_1$\emph{)}, $\mu(\textup{\texttt{Pn}})$ = (\textup{\texttt{n,P}}), $\mu_1(\textup{\texttt{Pn}})$~=~(\textup{\texttt{n,AP}}) and $removeAnn($\texttt{AP}$)$ = \texttt{P} for all \texttt{Pn} $\in$ \emph{dom(}$\mu$\emph{)}. (Note \texttt{Pn} could be \texttt{Cn} here).
\end{defn}

\begin{defn}{} \label{def-while}
Let $\beta$ be a while environment, and $\beta_1$ be an annotated while environment. We have that $\beta$ is \emph{equivalent} to $\beta_1$, written $\beta \approx_W \beta_1$, if and only if \emph{dom(}$\beta$\emph{)} = \emph{dom(}$\beta_1$\emph{)}, $\beta(\textup{\texttt{Wn}})$ = \textup{\texttt{P}}, $\beta_1(\textup{\texttt{Wn}})$ = \textup{\texttt{AP}} and $removeAnn($\texttt{AP}$)$ = \texttt{P} for all \texttt{Wn} $\in$ \emph{dom(}$\beta$\emph{)}.
\end{defn}

\begin{defn}{} \label{def-aux}
Let $\delta$ be an auxiliary store, and $\delta_1$ be an annotated auxiliary store. Firstly, we define the equivalence of stacks \texttt{St} and \texttt{St$'$}, written as \texttt{St} $\approx_{ST}$ \texttt{St$'$}, as true if both stacks have matching elements.  We have that $\delta$ is \emph{equivalent} to $\delta_1$, written $\delta \approx_A \delta_1$, if and only if for each stack \texttt{St} $\in$ \texttt{$dom(\delta)$}, we have $\delta($\texttt{St}$) \approx_{ST} \delta_1(\texttt{ST})$.
\end{defn}

\begin{defn}{} \label{def-eq-en}
Let $\square$ represent the set of environments \{$\sigma$,$\gamma$,$\mu$,$\beta$\} and $\square_1$ represent the set of annotated environments \{$\sigma_1$,$\gamma_1$,$\mu_1$,$\beta_1$\}. We have that $\square$ is \emph{equivalent} to $\square_1$, written $\square \approx \square_1$, if and only if $(\sigma,\gamma)~\approx_S~(\sigma_1,\gamma_1)$, $\mu \approx_P \mu_1$ and $\beta \approx_W \beta_1$.
\end{defn}

We now present our results. Theorem \ref{res-ann-1} states that identifiers are used in ascending order. Theorem~\ref{res-ann-2} below states if an original program terminates, the annotated version will also, and annotation does not change the behaviour of the program w.r.t the stores $\square$, but does produce a populated auxiliary store $\delta'$. 

\begin{theorem} \label{res-ann-1}
Let \textup{\texttt{AP}} and \textup{\texttt{AQ}} be annotated programs, \textup{\texttt{$\square$}} be the set of environments and $\delta$ be an auxiliary store. If $ (\textup{\texttt{AP} $\mid$ $\square$, $\delta$}) \tran{\circ}^* (\textup{\texttt{AP$'$} $\mid$ $\square'$, $\delta'$}) \tran{\text{n}} (\textup{\texttt{AQ} $\mid$ $\square''$, $\delta''$}) \rightarrow^* (\textup{\texttt{AQ$'$} $\mid$ $\square'''$, $\delta'''$}) \tran{\text{m}} (\textup{\texttt{AQ$''$} $\mid$ $\square''''$, $\delta''''$}) $, and the computation $ (\textup{\texttt{AQ} $\mid$ $\square''$, $\delta''$}) \rightarrow^* (\textup{\texttt{AQ$'$} $\mid$ $\square'''$, $\delta'''$}) $  does not have any identifier transitions, then \text{m} = \text{n} + \textup{\texttt{1}}.
\end{theorem}
\begin{proof}
The order of identifiers used during execution is maintained using \texttt{next()}. The program \texttt{AP} will begin with any number of steps. At some point, a transition occurs that will use the next available identifier \texttt{n}, while simultaneously incrementing \texttt{next()} by one to \texttt{m}. Any number of non m-rules can then apply, before the next m-action uses \texttt{next()} to get \texttt{m}. Hence, \text{m} = \text{n} + \texttt{1}. 
\end{proof}

\begin{lemma} \label{claim-1}
Let \textup{\texttt{P}} be a program, \textup{\texttt{$\square$}} be the set \{$\sigma$,$\gamma$,$\mu$,$\beta$\} of all environments, \textup{\texttt{$\square_1$}} be the set \{$\sigma_1$,$\gamma_1$,$\mu_1$,$\beta_1$\} of annotated environments such that $\square \approx \square_1$ and $\delta$ be the auxiliary store. If $ (\textup{\texttt{P} $\mid$ $\square$,$\delta$}) \hookrightarrow (\textup{\texttt{P$'$} $\mid$ $\square'$,$\delta$}) $, for some $\square'$, then there exists an execution $ (\textup{\texttt{$ann($P$)$} $\mid$ $\square_1$,$\delta$}) \tran{\circ} (\textup{\texttt{P$''$} $\mid$ $\square_1'$,$\delta'$}) $, for some $\square_1'$, $\delta'$ and \textup{\texttt{P$''$}}~=~$ann($\textup{\texttt{P}}$')$ such that $\square'$ $\approx$ $\square_1'$.
\end{lemma}

\begin{theorem} \label{res-ann-2}
Let \textup{\texttt{P}} be an original program, \textup{\texttt{$\square$}} be the set \{$\sigma$,$\gamma$,$\mu$,$\beta$\} of all environments, \textup{\texttt{$\square_1$}} be the set \{$\sigma_1$,$\gamma_1$,$\mu_1$,$\beta_1$\} of annotated environments such that $\square \approx \square_1$ and $\delta$ be the auxiliary store. 
If $ (\textup{\texttt{P} $\mid$ $\square$,$\delta$}) \hookrightarrow^* (\textup{\texttt{skip} $\mid$ $\square'$,$\delta$}) $, for some $\square'$, then there exists an execution $ (\textup{\texttt{$ann($P$)$} $\mid$ $\square_1$,$\delta$}) \tran{\circ}^* (\textup{\texttt{skip I} $\mid$ $\square_1'$,$\delta'$}) $, for some \textup{\texttt{I}}, $\square_1'$ and $\delta'$, such that $\square'$ $\approx$ $\square_1'$. 
\end{theorem}

\begin{proof}
The proof is by induction on the length of the sequence $(\texttt{P} \mid \square,\delta) \hookrightarrow^* (\texttt{skip} \mid \square',\delta')$. We consider \texttt{P} being either a sequential or parallel composition of programs. Both cases hold using Lemma \ref{claim-1}. 
\end{proof}
\noindent We note the implication holds in the opposite direction, but defer proof to future work.

\section{Reverse Semantics of Inverted Programs} \label{sec-inv}
Inversion is the process of generating the inverted version of a given program, produced from the executed annotated version, as the populated identifier stacks are necessary. As the inverted version is a program that executes forwards, inversion inverts the overall statement order. The functions that performs this, namely $inv()$ and the supplementary $i()$, are defined below. 

{\small $$ inv(\texttt{$\varepsilon$}) = \texttt{$\varepsilon$} \quad inv(\texttt{AS;AP}) = inv(\texttt{AP}); i(\texttt{AS}) \quad inv(\texttt{AP par AQ}) = inv(\texttt{AP}) \texttt{ par } inv(\texttt{AQ}) $$ }
\vspace{-.5cm}
{\small \begin{align*}
i(\texttt{skip I}) &= \texttt{skip I} \\
i(\texttt{X = e (pa,A)}) &= \texttt{X = e (pa,A)} \\
i(\texttt{if In b then AP else AQ end (pa,A)}) &= \texttt{if In b then $inv($AP$)$ else $inv($AQ$)$ end (pa,A)} \\
i(\texttt{while Wn b do AP end (pa,A)}) &= \texttt{while Wn b do $inv($AP$)$ end (pa,A)} \\
i(\texttt{begin Bn ADV ADP AP ARP ARV end}) &= \texttt{begin Bn $inv($ARV$)$ $inv($ARP$)$ $inv($AP$)$ $inv($ADP$)$} \\ &\phantom{=== } \texttt{$inv($ADV$)$ end} \\
i(\texttt{var X = v (pa,A)}) &= \texttt{remove X = v (pa,A)} \\
i(\texttt{proc Pn n is AP (pa,A)}) &= \texttt{remove Pn n is $inv($AP$)$ (pa,A)} \\
i(\texttt{call Cn n (pa,A)}) &= \texttt{call Cn n (pa,A)} \\
i(\texttt{remove X = v (pa,A)}) &= \texttt{var X = v (pa,A)} \\
i(\texttt{remove Pn n is AP (pa,A)}) &= \texttt{proc Pn n is $inv($AP$)$ (pa,A)}
\end{align*} }%

All inverted programs are of the annotated syntax in Section \ref{sec-ann}, but with \texttt{IP} and \texttt{IS} used for inverted programs and statements respectively. The inverse of \texttt{runB} and \texttt{runC} constructs simply invert the body.

Starting with the final state of all environments from the forwards execution, the inverted program will \emph{no longer perform any expression evaluation}, offering potential time saving when compared to traditional cyclic debugging. The result of any expression evaluation that happened during forwards execution is retrieved from the appropriate stack on $\delta$. For example, a while loop will iterate until the top element of stack \texttt{W} on $\delta$ is no longer true. The non-determinism that possibly occurred during the forwards execution will also not feature in the inverted execution. The identifiers assigned to statements ensure that any statement can only execute provided it has the highest unseen identifier. Using the function \texttt{previous()}, and starting it with the final value of \texttt{next()}, all m-rules will be reversed in \emph{backtracking order}. There is freedom in the order of non m-rules, with examples being parallel skip operations, or parallel block closings. Reversing these in any order produces no adverse effects. 

We recognise that all environments will now store inverted programs wherever necessary.  Since an inverted program is of the same syntax as an annotated program, our current environments are sufficient. We now return to our example discussed in Section \ref{sec-ann}. 

\begin{figure}[t]
  \begin{minipage}[b]{0.49\linewidth}
   \centering
   {\small \begin{lstlisting}[xleftmargin=2.0ex,mathescape=true]
$\texttt{begin b1}$    
 $\texttt{proc p1 fib is}$ 
  $\texttt{begin b2}$
   $\texttt{var T}$ = $\texttt{0 (b2*b1,A)}$;
   $\texttt{if i1}$ $\texttt{(N - 2 > 0)}$ $\texttt{then}$
    $\texttt{T}$ = $\texttt{F + S (b2*b1,A)}$;
    $\texttt{F}$ = $\texttt{S (b2*b1,A)}$;
    $\texttt{S}$ = $\texttt{T (b2*b1,A)}$;
    $\texttt{N}$ = $\texttt{N - 1 (b2*b1,A)}$;
    $\texttt{call c2.0 fib (b2*b1,A)}$;
   $\texttt{end  (b2*b1,A)}$
   $\texttt{remove T}$ = $\texttt{0 (b2,A)}$;
  $\texttt{end}$
 $\texttt{end  (b1,[1])}$ 
 $\texttt{call c1 fib (b1,[21])}$;
 $\texttt{remove p1 fib is AP (b1,[22])}$; 
$\texttt{end}$
\end{lstlisting} }
    \caption{Executed annotated version}
	\label{big-ex-ann}
  \end{minipage}
  \hspace{0.01cm}
  \begin{minipage}[b]{0.49\linewidth}
    \centering
    {\small \begin{lstlisting}[xleftmargin=2.0ex,mathescape=true]
$\texttt{begin b1}$    
 $\texttt{proc p1 fib is}$ 
  $\texttt{begin b2}$
   $\texttt{var T}$ = $\texttt{0 (b2*b1,A)}$;
   $\texttt{if i1}$ $\texttt{(N - 2 > 0)}$ $\texttt{then}$
    $\texttt{call c2 fib (b2*b1,A)}$;
    $\texttt{N}$ = $\texttt{N - 1 (b2*b1,A)}$;
    $\texttt{S}$ = $\texttt{T (b2*b1,A)}$;
    $\texttt{F}$ = $\texttt{S (b2*b1,A)}$;
    $\texttt{T}$ = $\texttt{F + S (b2*b1,A)}$;   
   $\texttt{end  (b2*b1,A)}$
   $\texttt{remove T}$ = $\texttt{0 (b2,A)}$;
  $\texttt{end}$
 $\texttt{end  (b1,[22])}$
 $\texttt{call c1 fib (b1,[21])}$; 
 $\texttt{remove p1 fib is IP (b1,[1])}$;  
$\texttt{end}$
	\end{lstlisting} }
    \caption{Inverted version}
    \label{big-ex-inv}
  \end{minipage}
\end{figure}

\begin{example} We now execute the annotated version of our original program (Figure \ref{big-ex-re-ann}), producing the final annotated version shown in Figure \ref{big-ex-ann}. The initial value of \texttt{next()} is 1. We enter block \texttt{b1}, and perform the procedure declaration, assigning the identifier 1 to its stack. The call statement then happens, performing a renamed copy of the procedure body. This will execute lines 3-9 using identifiers 2-6, before hitting the recursive call. This call will then execute another renamed copy of the procedure body. This renamed copy for the second call is shown in Figure \ref{big-ex-ann-call}, where the unique call name \texttt{c1:c2} has been used to rename all constructs. Lines 1-7 are executed using identifiers 7-11, before again hitting a recursive call. This renamed version is then executed, with the conditional on line 5 evaluating to false, meaning the recursion is now finished. This version uses the identifiers 12-14, before the recursive calls begin to close.  The second call (Figure \ref{big-ex-ann-call}) then concludes, executing lines 8-11 using identifiers 15-17. The first call then concludes using the identifiers 18-20. Finally, the original program concludes, using identifier 21 to finish the call statement, and then the last identifier 22 (meaning \texttt{next()} = 23) to remove the procedure declaration. This concludes our execution, producing the final state \texttt{F=7}, \texttt{S=11} and \texttt{N=2}. 

We now consider the inverted execution. Application of the function $inv()$ to the executed annotated version (Figure \ref{big-ex-ann}) produces the inverted version shown in Figure \ref{big-ex-inv}. The initial value of \texttt{previous()} will be 22 (\texttt{next()} - 1). Inverse execution begins by opening the outer block and performing the procedure declaration (the inverse of the procedure removal) using identifier 22. The call statement is then entered using identifier 21. A renamed copy of the procedure body is then executed, performing lines 3-6 using identifiers 20-18. The recursive call is then hit, beginning the execution of another renamed copy, shown in Figure \ref{big-ex-inv-call}. Lines 1-4 are executed using identifiers 17-15, before the recursive call is hit again. The third renamed version now executes, and since the condition will be false (retrieved from the stack \texttt{B} on $\delta$), recursion now stops. This call concludes using identifiers 14-12, before the second call (Figure~\ref{big-ex-inv-call}) now concludes lines 5-11 using identifiers 11-7. The first call then concludes using identifiers 6-2. Finally, the original program concludes with the removal of the procedure (inverse of declaration) using identifier~1. This execution order restores the initial program state.  
 \end{example}  \label{ex-2}

\begin{figure}[t]
  \begin{minipage}[b]{0.49\linewidth}
   \centering
   {\small \begin{lstlisting}[xleftmargin=2.0ex,mathescape=true]
$\texttt{begin c1:c2:b2}$
 $\texttt{var T}$ = $\texttt{0 (c1:c2:b2*b1,[7])}$;
 $\texttt{if c1:c2:i1}$ $\texttt{(N - 2 > 0)}$ $\texttt{then}$
  $\texttt{T}$ = $\texttt{F + S (c1:c2:b2*b1,[8])}$;    
  $\texttt{F}$ = $\texttt{S (c1:c2:b2*b1,[9])}$;   
  $\texttt{S}$ = $\texttt{T (c1:c2:b2*b1,[10])}$;   
  $\texttt{N}$ = $\texttt{N - 1 (c1:c2:b2*b1,[11])}$;   
  $\texttt{call c1:c2:c2 fib (c1:c2:b2*b1,[15])}$;    
 $\texttt{end  (c1:c2:b2*b1,[16])}$
 $\texttt{remove T}$ = $\texttt{0 (c1:c2:b2*b1,[17])}$;
$\texttt{end}$
\end{lstlisting} }
    \caption{Executed annotated version of 2nd call}
	\label{big-ex-ann-call}
  \end{minipage}
  \hspace{0.01cm}
  \begin{minipage}[b]{0.49\linewidth}
    \centering
    {\small \begin{lstlisting}[xleftmargin=2.0ex,mathescape=true]
$\texttt{begin c1:c2:b2}$
 $\texttt{var T}$ = $\texttt{0 (c1:c2:b2*b1,[17])}$;
 $\texttt{if c1:c2:i1}$ $\texttt{(N - 2 > 0)}$ $\texttt{then}$
  $\texttt{call c1:c2:c2 fib (c1:c2:b2*b1,[15])}$;   
  $\texttt{N}$ = $\texttt{N - 1 (c1:c2:b2*b1,[11])}$;     
  $\texttt{S}$ = $\texttt{T (c1:c2:b2*b1,[10])}$;   
  $\texttt{F}$ = $\texttt{S (c1:c2:b2*b1,[9])}$;     
  $\texttt{T}$ = $\texttt{F + S (c1:c2:b2*b1,[8])}$;      
 $\texttt{end  (c1:c2:b2*b1,[16])}$
 $\texttt{remove T}$ = $\texttt{0 (c1:c2:b2*b1,[7])}$;
$\texttt{end}$
	\end{lstlisting} }
    \caption{Inverted version of 2nd call}
    \label{big-ex-inv-call}
  \end{minipage}
\end{figure}

Prior to defining the inverse operational semantics, we must introduce the function \emph{setAI()}. This takes the output of the function \emph{getAI()} from Section \ref{sec-ann} and a program, and returns a copy of this program with the given annotation information inserted. Recall the functions \emph{IreP} and \emph{IreL} from Section \ref{ssec-scope}.

We now give the operational semantics of inverted programs. We introduce \emph{reverse m-rules}, the name given to all statements of the inverse execution that use an identifier. Transitions $\revtran{m}$ are also called reverse identifier transitions. All other rules remain non m-rules. We note a correspondence between each m-rule and the matching reverse m-rule. 


\vspace{.1cm}
\noindent \textbf{Sequential and Parallel Composition} The inverted program is still executed forwards, meaning these are like those in Section \ref{sec-ann}, but with $\rightsquigarrow$ replacing $\rightarrow$, and \texttt{IP} and \texttt{IS} replacing \texttt{AP} and \texttt{AS} respectively. Each rule is named correspondingly to Section \ref{sec-ann}, but with the appended `\texttt{a}' replaced with `\texttt{r}'. For example, rule \text{[S1a]} is now \text{[S1r]}.

\vspace{.1cm}
\noindent \textbf{Assignment} The inverse of an assignment will be a reverse m-rule, allowed to execute provided the top element of stack \texttt{A} is \texttt{m} (written \texttt{A = m:A$'$}), and \texttt{m} is the last used identifier (via \texttt{previous()}). This retrieves the old value, with matching identifier \texttt{m}, from the stack on $\delta$ for this variable, and assigns it to the corresponding location (evaluated as in Section \ref{sec-ann}).

\vspace{-.5cm}
{\small \begin{align*}
&\text{[D1r]} \quad \frac{\texttt{A = m:A$'$} \quad \texttt{m = previous()} \quad  \texttt{$evalVar(\gamma$,pa,X$)$ = l} \quad \texttt{$\delta($X$)$ = (m,v):X$'$}}{(\texttt{X = e (pa,A)} \mid \delta,\sigma,\square) \revtran{m} (\texttt{skip A$'$} \mid \delta[\texttt{X/X$'$}],\sigma[\texttt{l} \mapsto \texttt{v}],\square)}
\end{align*} }%

\vspace{.1cm}
\noindent \textbf{Conditional} This will begin with the reverse m-rule \text{[I1r]} that, provided this statement has the next identifier to invert (via \texttt{previous()} and \texttt{A = m:A$'$}) retrieves the boolean value from the stack \texttt{B} with matching identifier \texttt{m} on $\delta$. Rules \text{[I2r]} and \text{[I3r]} follow Section \ref{sec-ann}, but with inverted programs. Finally, \text{[I4r]} and \text{[I5r]} simply concludes the statement.

\vspace{-.5cm}
{\small \begin{align*}
&\text{[I1r]} \quad \frac{\texttt{A = m:A$'$} \quad \texttt{m = previous()} \quad \texttt{$\delta($B$)$ = (m,V):B$'$}}{(\texttt{S} \mid \delta,\square) \revtran{m} (\texttt{if In V then IP else IQ end (pa,A$'$)} \mid \delta[\texttt{B/B$'$}],\square)} \\[4pt] & \phantom{\text{[I1r]} \quad } \text{where } \texttt{S} = \texttt{if In b then IP else IQ end (pa,A)} \\[6pt]
&\text{[I2r]} \quad \frac{(\texttt{IP} \mid \square) \revtran{\circ} (\texttt{IP$'$} \mid \square')}{(\texttt{if In T then IP else IQ end (pa,A)} \mid \square) \revtran{\circ} (\texttt{if In T then IP$'$ else IQ end (pa,A)} \mid \square')} \\[6pt]
&\text{[I3r]} \quad \frac{(\texttt{IQ} \mid \square) \revtran{\circ} (\texttt{IQ$'$} \mid \square')}{(\texttt{if In F then IP else IQ end (pa,A)},\square) \revtran{\circ} (\texttt{if In F then IP else IQ$'$ end (pa,A)} \mid \square')} \\[6pt]
&\text{[I4r]} \quad \frac{}{(\texttt{if In T then skip I else IQ end (pa,A)} \mid \square) \rightsquigarrow (\texttt{skip} \mid \square)} \\[6pt]
&\text{[I5r]} \quad \frac{}{(\texttt{if In F then IP else skip I end (pa,A)} \mid \square) \rightsquigarrow (\texttt{skip} \mid \square)} 
\end{align*} }%

\vspace{.5cm}
\noindent \textbf{While Loop} The reverse m-rule \text{[W1r]} handles the first iteration of a loop, retrieving either the \texttt{T} or \texttt{F} from the stack \texttt{B} on $\delta$, and creating a mapping on $\beta$. This mapping is similar to that of \text{[W1a]} but with an inverted copy of the loop \texttt{IP$'$} that is both renamed and updated with annotated information \texttt{C}, retrieved from the stack \texttt{WI} (via $setAI()$). The reverse m-rule \text{[W2r]} handles all iterations except the first, meaning the renaming is applied to the current mapping, written $\beta$[\texttt{Wn} $\Rightarrow$ \texttt{IR$'$}]. Both rules only execute provided they have the last used identifier (via \texttt{previous()} and \texttt{A = m:A$'$}). The body is then executed repeatedly with changes reflected to the stored copy, written $\beta'$[\texttt{Wn} $\Rightarrow$ \texttt{IR$'$}], via rule \text{[W3r]}. The loop continues through the rule \text{[W4r]}, until the condition is false and thus the mapping removed by rule \text{[W5r]}. 

\vspace{-.3cm}
{\small \begin{align*}
&\text{[W1r]} \quad \frac{\texttt{m = previous()} \quad \texttt{A = m:A$'$} \quad \texttt{$\beta($Wn$)$ = }\emph{und} \quad \texttt{$\delta($WI$)$=(m,C):WI$'$} \quad \texttt{IP$'$ = $IreL($setAI(IP,C)$)$}}{(\texttt{S} \mid \delta,\beta,\square) \revtran{m} (\texttt{while Wn b do IP$'$ end (pa,A$'$)} \mid \delta[\texttt{WI/WI$'$}],\beta[\texttt{Wn} \Rightarrow \texttt{IR}],\square)} \\[4pt] & \phantom{\text{[W1r]} \quad } \text{where } \texttt{S} = \texttt{while Wn b do IP end (pa,A)} \text{and } \texttt{IR} = \texttt{while Wn b do IP$'$ end (pa,A$'$)}  \\[6pt]
&\text{[W2r]} \quad \frac{\texttt{m = previous()} \quad \texttt{A = m:A$'$} \quad \texttt{$\beta($Wn$)$ = while Wn b do IQ end (pa,A)} \quad \texttt{$\delta($W$)$ = (m,V):W$'$}}{(\texttt{S} \mid \delta,\beta,\square) \revtran{m} (\texttt{while Wn V do $IreL($IQ$)$ end (pa,A$'$)} \mid \delta[\texttt{W/W$'$}],\beta[\texttt{Wn} \Rightarrow \texttt{IR$'$}],\square)} \\[4pt] & \phantom{\text{[W2r]} \quad } \text{where } \texttt{S} = \texttt{while Wn b do IP end (pa,A)} \text{ and } \texttt{IR$'$} = \texttt{while Wn b do $IreL($IQ$)$ end (pa,A$'$)}  \\[6pt]
&\text{[W3r]} \quad \frac{(\texttt{IR} \mid \delta,\beta,\square) \revtran{\circ} (\texttt{IR$'$} \mid \delta',\beta',\square')}{(\texttt{while Wn T do IR end (pa,A)} \mid \delta,\beta,\square) \revtran{\circ} (\texttt{while Wn T do IR$'$ end (pa,A)} \mid \delta',\beta'',\square')} \\[4pt] & \phantom{\text{[W3r]} \quad } \text{where } \texttt{$\beta''$} = \texttt{$\beta'$}[\emph{refW(}\texttt{Wn,IR$'$}\emph{)}] \\[6pt]
&\text{[W4r]} \quad \frac{\texttt{$\beta($Wn$)$ = IP} }{(\texttt{while Wn T do skip I end (pa,A)} \mid \beta,\square) \rightsquigarrow (\texttt{IP} \mid \beta,\square)}  \\[6pt]
&\text{[W5r]} \quad \frac{\texttt{$\beta($Wn$)$ = IR}}{(\texttt{while Wn F do IP end (pa,A)} \mid \beta,\square) \rightsquigarrow (\texttt{skip} \mid \beta[\texttt{Wn}],\square)}  
\end{align*} }%

\vspace{.1cm}
\noindent \textbf{Block} The inversion of a block is very similar to that of Section \ref{sec-ann}, but with the inverted syntax. Recall that a variable declaration in the inverted syntax is the inverse of a variable removal, and similarly for procedures. This means an inverted block has a body of the form shown below as \texttt{IP} in \text{[B1r]}.

\vspace{-.5cm}
{\small \begin{align*}
&\text{[B1r]} \quad \frac{}{(\texttt{begin Bn IP end} \mid \square) \rightsquigarrow (\texttt{runB IP end} \mid \square)} \quad \text{where } \texttt{IP } \text{=} \texttt{ IDV;IDP;IQ;IRP;IRV} \\[6pt] 
&\text{[B2r]} \quad \frac{(\texttt{IP} \mid \square) \revtran{\circ} (\texttt{IP$'$} \mid \square')}{(\texttt{runB IP end} \mid \square) \revtran{\circ} (\texttt{runB IP$'$ end} \mid \square')}  \quad \text{[B3r]} \quad \frac{}{(\texttt{runB skip I end} \mid \square) \rightsquigarrow (\texttt{skip} \mid \square)}  
\end{align*} }%

\vspace{.1cm}
\noindent \textbf{Variable and Procedure Declaration} Reverse variable declaration \text{[L1r]} is a reverse m-rule, allowed to execute provided it has the last used identifier (via \texttt{previous()} and \texttt{A = m:A$'$}). The given value is ignored, and the variable is instead set to its final value retrieved from the corresponding stack on $\delta$ (written \texttt{$\delta($X$)$} \texttt{=} \texttt{(m,v$'$):X$'$}). The location \texttt{l} is used as in Section \ref{sec-ann}. A procedure declaration \text{[L2r]} creates the basis mapping as in Section \ref{sec-ann}, but with inverted programs, provided its identifiers allow this.

\vspace{-.5cm}
{\small \begin{align*}
&\text{[L1r]} \quad \frac{\texttt{A = m:A$'$} \quad \texttt{m = previous()} \quad \texttt{$\delta($X$)$ = (m,v$'$):X$'$} \quad \texttt{nextLoc() = l} \quad \texttt{pa} = \texttt{Bn*pa$'$} }{(\texttt{var X = v (pa,A)} \mid \delta,\sigma,\gamma,\square) \revtran{m} (\texttt{skip A$'$} \mid \delta[\texttt{X/X$'$}],\sigma[\texttt{l} \mapsto \texttt{v$'$}],\gamma[(\texttt{X},\texttt{Bn}) \Rightarrow \texttt{l}],\square)} \\[6pt]
&\text{[L2r]} \quad \frac{\texttt{A = m:A$'$} \quad \texttt{m = previous()}}{(\texttt{proc Pn n is IP (pa,A)} \mid \mu,\square) \revtran{m} (\texttt{skip A$'$} \mid \mu[\texttt{Pn} \Rightarrow \texttt{(n,IP)}],\square)}
\end{align*} }%

\vspace{.1cm}
\noindent \textbf{Procedure Call} The reverse m-rule \text{[G1r]} creates a renamed copy \texttt{IP$'$} of the basis procedure body \texttt{IP}, that has annotated changes \texttt{C} from stack \texttt{WI} inserted. This is inserted into $\mu$ via $\mu$[\texttt{Cn} $\Rightarrow$ \texttt{(n,IP$'$)}]. Rule \text{[G2r]} follows Section \ref{sec-ann} but uses inverted programs, while \text{[G3r]} removes the mapping.

\vspace{-.5cm}
{\small \begin{align*}
&\text{[G1r]} \quad \frac{\texttt{m = previous()} \quad \texttt{A = m:A$'$} \quad \texttt{$\mu(evalP($n,pa$)$) = (n,IP)} \quad \texttt{$\delta($Pr$)$ = (m,C):Pr$'$} }{(\texttt{call Cn n (pa,A)} \mid \mu,\square) \revtran{m} (\texttt{runC Cn IP$'$ end A$'$} \mid \mu[\texttt{Cn} \Rightarrow \texttt{(n,IP$'$)}],\square)} \\&\phantom{[RG1] \quad} \text{where } \texttt{IP$'$} = \texttt{$IreP(setAI($IP,C$)$,Cn$)$}  \\[6pt]
&\text{[G2r]} \quad \frac{(\texttt{IP} \mid \mu,\square) \revtran{\circ} (\texttt{IP$'$} \mid \mu',\square')}{(\texttt{runC Cn IP end A} \mid \mu,\square) \revtran{\circ} (\texttt{runC Cn IP$'$ end A} \mid \mu'[\emph{refC(}\texttt{Cn},\texttt{IP$'$}\emph{)}],\square')} \\[6pt]
&\text{[G3r]} \quad \frac{\texttt{$\mu$(Cn) = IP}}{(\texttt{runC Cn skip I end A} \mid \mu,\square) \rightsquigarrow (\texttt{skip A} \mid \mu[\texttt{Cn}],\square)}
\end{align*} }%

\vspace{.1cm}
\noindent \textbf{Variable and Procedure Removal} Variable removal \text{[H1r]} is similar to Section \ref{sec-trad}, as no state-saving is required. But it is a reverse m-rule and can execute provided \texttt{A = m:A$'$} and \texttt{m = previous()}. Procedure removal \text{[H2r]} is a reverse m-rule with no state-saving, removing the mapping.

\vspace{-.5cm}
{\small \begin{align*}
&\text{[H1r]} \quad \frac{\texttt{A = m:A$'$} \quad \texttt{m = previous()} \quad \texttt{pa} = \texttt{Bn*pa$'$} \quad \texttt{$\gamma($X,Bn$)$ = l}}{(\texttt{remove X = v (pa,A)} \mid \delta,\sigma,\gamma,\square) \revtran{m} (\texttt{skip A$'$} \mid \delta,\sigma[\texttt{l} \mapsto \texttt{0}],\gamma[(\texttt{X},\texttt{Bn})],\square)} \\[2pt]
&\text{[H2r]} \quad \frac{\texttt{A = m:A$'$} \quad \texttt{m = previous()} \quad \texttt{$\mu$(Pn) = IQ}}{(\texttt{remove Pn n is IP (pa,A)} \mid \mu,\square) \revtran{m} (\texttt{skip A$'$} \mid \mu[\texttt{Pn}],\square)}
\end{align*} }%

\subsection{Results} \label{ssec-inv-res}
Prior to describing our inversion results, we first note that the definitions in Section \ref{ssec-ann-res} are now used to relate annotated environments with inverted environments, instead of traditional. As an example, Definition~\ref{def-proc} would no longer have that \texttt{P} = $removeAnn($\texttt{AP}$)$, but instead that \texttt{IP} = $inv($\texttt{P}$)$. 

Theorem \ref{res-inv-1} states that identifiers are used in descending order throughout a reverse execution (opposite of Theorem \ref{res-ann-1}). Theorem \ref{res-inv-2} shows that if an original \emph{sequential} program and its annotated execution terminate, then the reverse execution will also, and that the reverse execution beginning in the final state can restore the initial state. 

\begin{theorem} \label{res-inv-1}
Let \textup{\texttt{P}} and \textup{\texttt{Q}} be original programs, \textup{\texttt{AP}} and \textup{\texttt{AQ}} be the annotated versions producing the executed versions \textup{\texttt{AP$'$}} and \textup{\texttt{AQ$'$}} respectively, and \textup{\texttt{IP}} and \textup{\texttt{IQ}} be the inverted versions \textup{\texttt{$inv($AP$')$}} and \textup{\texttt{$inv($AQ$')$}} respectively. Further let \textup{\texttt{$\square$}} be the set of all environments and $\delta$ be the auxiliary store. If $ (\textup{\texttt{IP} $\mid$ $\square$, $\delta$}) \revtran{\circ}^* (\textup{\texttt{IP$'$} $\mid $ $\square'$, $\delta'$}) \revtran{\textup{\text{n}}} (\textup{\texttt{IQ} $\mid$ $\square''$, $\delta''$}) \rightsquigarrow^* (\textup{\texttt{IQ$'$} $\mid$ $\square'''$, $\delta'''$}) \revtran{\textup{\text{m}}} (\textup{\texttt{IQ$''$} $\mid$ $\square''''$, $\delta''''$}) $, and provided that the computation $ (\textup{\texttt{IQ} $\mid$ $\square''$, $\delta''$}) \rightsquigarrow^* (\textup{\texttt{IQ$'$} $\mid$ $\square'''$, $\delta'''$}) $  does not include any identifier transitions, then \textup{\text{m}} = \textup{\text{n}} - \textup{\texttt{1}}.
\end{theorem}
\begin{proof} The order of identifiers is maintained using the function \texttt{previous()}. This proof follows closely to the argument within the proof of Theorem \ref{res-ann-1}, but uses $\revtran{\circ}$ in place of $\tran{\circ}$. 
\end{proof}

\begin{theorem} \label{res-inv-2}
Let \textup{\texttt{P}} be a sequential program (does not contain \textup{\texttt{par}}) and  \textup{\texttt{AP}} be \textup{\texttt{$ann($P$)$}}. Further let \textup{\texttt{$\square$}} be the set \{$\sigma$,$\gamma$,$\mu$,$\beta$\} of all environments, \textup{\texttt{$\square_1$}} be the set \{$\sigma_1$,$\gamma_1$,$\mu_1$,$\beta_1$\} of annotated environments such that $\square \approx \square_1$, \textup{\texttt{$\square_1'$}} be the set \{$\sigma_1'$,$\gamma_1'$,$\mu_1'$,$\beta_1'$\} of final annotated environments, \textup{\texttt{$\square_2$}} be the set \{$\sigma_2$,$\gamma_2$,$\mu_2$,$\beta_2$\} of inverted environments such that $\square_2 \approx \square_1'$, $\delta$ be the auxiliary store, $\delta'$ be the final auxiliary store and $\delta_2$ be the inverted auxiliary store such that $\delta_2 \approx_A \delta'$.

\begin{enumerate}[ 1{)} ]
\item If $ (\textup{\texttt{P} $\mid$ $\square$,$\delta$}) \hookrightarrow^* (\textup{\texttt{skip} $\mid$ $\square'$,$\delta$}) $, for some $\square'$, and there exists an annotated execution $ (\textup{\texttt{AP} $\mid$ $\square_1$,$\delta$}) \tran{\circ}^*  (\textup{\texttt{skip I} $\mid$ $\square_1'$,$\delta'$}) $, for some \textup{\texttt{I}}, $\square_1'$ and $\delta'$ such that the executed annotated version of \textup{\texttt{AP}} produced by its execution is \textup{\texttt{AP$'$}}, then there also exists $ (\textup{\texttt{IP} $\mid$ $\square_2$,$\delta_2$}) \revtran{\circ}^* (\textup{\texttt{skip I$'$} $\mid$ $\square_2'$,$\delta_2'$}) $, for \textup{\texttt{IP}} = \textup{\texttt{$inv($AP$')$}} and some \textup{\texttt{I$'$}}, $\square_2'$ and $\delta_2'$. (Termination)
\item If $ (\textup{\texttt{P} $\mid$ $\square$,$\delta$}) \hookrightarrow^* (\textup{\texttt{skip} $\mid$ $\square'$,$\delta$}) $, for some $\square'$, and there exists an annotated execution $ (\textup{\texttt{AP} $\mid$ $\square_1$,$\delta$}) \tran{\circ}^* (\textup{\texttt{skip I} $\mid$ $\square_1'$,$\delta'$}) $, for some \textup{\texttt{I}}, $\square_1'$ and $\delta'$, such that $\square' \approx \square_1'$ and that the executed annotated version of \textup{\texttt{AP}} produced by its execution is \textup{\texttt{AP$'$}}, then there also exists $ (\textup{\texttt{IP} $\mid$ $\square_2$,$\delta_2$}) \revtran{\circ}^* (\textup{\texttt{skip I$'$} $\mid$ $\square_2'$,$\delta_2'$}) $, for \textup{\texttt{IP}} = \textup{\texttt{$inv($AP$')$}} and some  \textup{\texttt{I$'$}}, $\square_2'$ and $\delta_2'$, such that $\square_2' \approx \square$ and $\delta_2' \approx_A \delta$.
\end{enumerate}
\end{theorem}

\begin{proof}
This proof is by induction on the length of the sequence $(\texttt{P} \mid \square,\delta) \hookrightarrow^* (\texttt{skip} \mid \square',\delta')$. 
\end{proof}

\noindent The full version of Theorem \ref{res-inv-2}, where programs can contain parallel composition, is currently being considered. We note the implication in the other direction would be valid, but defer proof to future work.

\section{Conclusion} \label{sec:conc}
We have presented an approach to reversing a language containing blocks, local variables, procedures and the interleaving parallel composition. We defined annotation, the process of creating a state-saving annotated version capable of assigning identifiers to capture the interleaving order. This was proved to not alter the behaviour of the original program and to populate the auxiliary store. Inversion creates an inverted version that uses this saved information to restore the initial program state, with this being proved to hold. The auxiliary store is also restored, meaning it is garbage free. 

We are also currently developing a simulator capable of implementing this approach, with one application being to aid the proof of our results. The current prototype is capable of simulating the annotated forwards execution, with the simulation of the inverted execution currently being worked on. This will be used to evaluate the performance overhead and costs associated with our approach to reversibility.

Our future work will continue the development of this simulator, as well as modify the approach allowing for causal-consistent reversibility.

\subsubsection*{Acknowledgements}
We are grateful to the referees for their detailed and helpful comments and suggestions. The authors acknowledge partial support of COST Action IC1405 on Reversible Computation - extending horizons of computing. The third author is supported by JSPS KAKENHI grant numbers 17H01722 and 17K19969.

\nocite{*}
\bibliographystyle{eptcs}

\end{document}